\documentclass[runningheads]{llncs}
\def\cameraReady{} 

\usepackage[T1]{fontenc}
\usepackage{graphicx}
\usepackage{amsmath}
\usepackage{amssymb}
\usepackage{dirtytalk}
\usepackage{apxproof}
\usepackage{spalign}
\usepackage[dvipsnames]{xcolor}
\usepackage{url}
\usepackage{hyperref}
\usepackage[hyphenbreaks]{breakurl}

\usepackage{amsmath}
\usepackage{xspace}
\usepackage{mathtools}
\usepackage{algorithm}
\usepackage[noend]{algpseudocode}
\usepackage{amsmath}

\newtheorem{assumption}{Assumption}

\begin{document}
\title{On Fair Ordering and Differential Privacy}
%
%

\ifdefined\cameraReady
\author{Shir Cohen\inst{1} \and
Neel Basu\inst{2} \and
Soumya Basu\inst{2} \and
Lorenzo Alvisi\inst{1}}
\institute{Cornell University \and
Nuveaux Trading}
\else
\author{}{}
\institute{}
\fi

\authorrunning{S. Cohen et al.}
%
%
\maketitle              
\begin{abstract}

In blockchain systems, fair transaction ordering is crucial for a trusted and regulation-compliant economic ecosystem. Unlike traditional State Machine Replication (SMR) systems, which focus solely on liveness and safety, blockchain systems also require a fairness property. This paper examines these properties and aims to eliminate algorithmic bias in transaction ordering services.

We build on the notion of equal opportunity. We characterize transactions in terms of relevant and irrelevant features, requiring that the order be determined solely by the relevant ones. Specifically,  transactions with identical relevant features should have an equal chance of being ordered before one another. We extend this framework to define a property where the greater the distance in relevant features between transactions, the higher the probability of prioritizing one over the other.

We reveal a surprising link between equal opportunity in SMR and Differential Privacy (DP), showing that any DP mechanism can be used to ensure fairness in SMR. This connection not only enhances our understanding of the interplay between privacy and fairness in distributed computing but also opens up new opportunities for designing fair distributed protocols using well-established DP techniques.

\end{abstract}

%
%
%


\section{Introduction}

State Machine Replication (SMR)~\cite{lamport1978time} is a technique that enables a set of fault-prone deterministic servers to emulate a single, fault-free deterministic state machine by agreeing upon, and then processing, the same, totally-ordered sequence of client requests. SMR is a foundational approach to building fault-tolerant distributed systems, and  SMR implementations have existed for decades~\cite{schneider1990implementing}.  
Recently, however, SMR's role in supporting blockchain applications has brought it new attention.

This paper focuses on the demands than an  especially important new SMR application -- cryptocurrencies -- places on SMR implementations:  specifically, on the need for these implementation to guarantee some form of fairness in  total order of requests that they produce. Fairness is included in the correctness specification of traditional SMR.  When SMR is used solely for fault tolerance, the specific total order agreed upon is irrelevant; however, when client requests,  or \emph{transactions}, consist of currency transfers between accounts or smart contract operations, the order clearly matters~\cite{zhang2020byzantine}.
Recognizing this, a recent line of work has begun addressing the challenge of fair ordering by proposing different application-specific properties~\cite{cachin2022quick,kelkar2020order,kursawe2020wendy,ramseyer2024brief,zhang2020byzantine}.

Our starting point for this paper is a recent attempt at a  more systematic treatment of fair ordering~\cite{bercow}, which associates with every request a set of {\em features}. Some of these features are {\em relevant} to the criteria ostensibly used to determine the final ordering; others are not. Their goal is to eliminate from the ordering process biases that may favor certain clients or requests unfairly ({\em i.e.}, on the basis of {\em irrelevant} features). Such biases may be systemic ({\em e.g.}, by favoring clients physically closer to the ordering service) or due to malicious parties attempting to manipulate the ordering to their advantage.

Zhang et al.~\cite{bercow} submit that a fair ordering should guarantee {\em impartiality}: intuitively requests with the same relevant features should have an equal chance of being prioritized. They formalize this intuition in a property they call $\epsilon$-Ordering Equality, which, using the absolute time at which a request is issued by a client as its sole relevant feature, guarantees impartiality up to some small parameter $\epsilon$. Building upon this work, our paper makes three contributions.

First, it extends $\epsilon$-Ordering Equality by introducing a refined definiton that ($i$) applies to arbitrary relevant features, instead of being tied solely to a request's time of issue; and ($ii$) gracefully degrades its impartiality guarantee as relevant features drift apart, rather than applying only to the case when relevant features have identical values.

Second, the paper reveals and formalizes a surprising connection between fairness properties and Differential Privacy (DP)~\cite{dwork2006differential}, a powerful framework for protecting individual privacy when sharing statistical information about a dataset. Our key insight is  that the very techniques used to ``hide'' information for privacy can also be leveraged to conceal it from the ordering service, enabling a new method to achieve a high level of fairness in the system.

Finally, we show how enforcing blockchain fairness via our refined notion of $\epsilon$-Ordering Equality can not only reduce front-running and sandwich attacks, as first noted by Zhang et al.~\cite{bercow}, but can also be applied to miners' fees to mitigate other
Miner Extractable Value issues~\cite{daian2020flash}.

We conclude by introducing several open questions, laying the foundation for further research on the subject.

\section{Model}



We consider a message-passing model with an unbounded set of clients. A client is \emph{crashed} if it halts prematurely at some point during an execution. If it deviates from the protocol it is \emph{Byzantine}. We do not impose limitations on the number of faults in the system.
Additionally, we assume the existence of a  reliable ordering server $\mathcal{S}$. In our system model, $\mathcal{S}$ can be implemented using one of several existing solutions for building Byzantine fault tolerant replicated services~\cite{castro1999practical,lamport2001paxos,yin2019hotstuff}

Clients send requests (transactions) to $\mathcal{S}$ over reliable channels; in turn, $\mathcal{S}$ outputs a total order over the requests. We denote by $r\prec r'$ the fact that $r$ appears before $r'$ in the total order. Following~\cite{schneider1990implementing}, we say that a request $r$ is \emph{stable} once no other request with a lower unique identifier can be subsequently delivered to $\mathcal{S}$. It is up to the implementations of $\mathcal{S}$ to devise a favorable stability test for incoming requests; the literature offers  several examples~\cite{schneider1990implementing}.

A client request $r$ is associated with a set of features: these may  include, for example,  information identifying the client that issued $r$, the request's type, or the time that it was received by $\mathcal{S}$. Some of these features are deemed relevant to determining $r$'s place in the total order produced by $\mathcal{S}$; the rest is irrelevant~\cite{bercow}.
We say that $r$ and $r'$ are \emph{adjacent} for a given set of relevant features, denoted as $r \sim r'$, if and only if the values assumed by the relevant features in $r$ and $r'$ are the same.




\section{Fair Ordering}

To achieve fair ordering, one must define what fairness means.
Many notions exist in the literature~\cite{cachin2022quick,kelkar2020order,kursawe2020wendy,ramseyer2024brief,zhang2020byzantine}; we adopt the scheme proposed by~\cite{bercow} because we find that it offers a general structure for reasoning about fairness.
In this framework, fair ordering is based only on the value assumed by the features that the ordering service considers relevant. 
Among requests sharing the same values for these features, equity is sought, meaning they have equal \emph{chances} of being ordered first. 
This requirement is specified via the  $\epsilon$-Ordering Equality property, which~\cite{bercow} defines in terms of a specific relevant feature -- the time at which a request is invoked --  as follows~\cite{bercow}:

\begin{definition}[$\epsilon$-Ordering Equality]
For any two requests $r,r'$ invoked at the same time, 
$| Pr[r\prec r']-Pr[r'\prec r]  | \leq f(\epsilon)$, for some function $f$.


\end{definition}

We generalize this definition in two ways. First, we strengthen it so that it covers relevant features other than just the request's time of invocation. 
Second, we adjust the definition so that the relationship between probabilities is bounded by a multiplicative factor of $e^\epsilon$ instead of an additive error.
This choice ensures that relative differences between probabilities are preserved, particularly when these probabilities are small. Our revised definition is:




\begin{definition}[$\epsilon$-Ordering Equality -- revised]
    For any two requests $r,r'$, if $r\sim r'$, then $Pr[r\prec r']\leq e^{\epsilon} Pr[r'\prec r]  $.

\end{definition}


To achieve $\epsilon$-Ordering Equality, the system must have access to the values assumed by a request's relevant features. But where should these values be measured? At first glance, it would seem that clients are better positioned to do so. For instance, consider the case of a system that aims  to order request according to the time of their  invocation, as in ~\cite{bercow}. The client is, of course, perfectly placed to measure the relevant feature and pass it on to $\mathcal{S}$ as metadata. $\mathcal{S}$ in turn would wait for requests with the same metadata to stabilize, and then assign them a random order.
Unfortunately, a Byzantine client may deliberately misreport the invocation time. 
While this does not change the true value of the relevant feature, it introduces irrelevant information that is indistinguishable from the actual data by $\mathcal{S}$.

The remaining alternative is to measure relevant features at  $\mathcal{S}$. This approach brings its own challenges, however: if $\mathcal{S}$ orders requests according to the time it {\em receives} them, then irrelevant features, such as network delays or the geolocation of the client with respect to $\mathcal{S}$ may alter, or outright reverse, the desired order.


To address this, we focus on how $\mathcal{S}$ can interpret incoming requests. Assume, for simplicity, a request has only one relevant feature -- in practice, multiple features can be aggregated as desired. The {\em score} that $\mathcal{S}$ assigns to a request will be a function of both its relevant and irrelevant features: $\emph{score}(r)=r.\emph{relev}+r.\eta$.







Here, {\em relev} holds the relevant information, and $\eta$ is the noise from irrelevant features that biases the score.
We focus on scenarios where the score is computed by summing  relevant and irrelevant features. A simple example is a score based on when a request is received, where consists of the sum of the invocation time (which is a relevant feature)  and the network delay experienced by the request to reach $\mathcal{S}$ (which we consider an irrelevant feature). Another example of an additive score is one based on transaction fees. In this case, the relevant feature is the fee for block transactions, while any kickbacks exposed by clients are irrelevant and should be ignored for fair ordering (see \S\ref{sec:DP_for_miners_fees} for more on these examples).
A fair ordering mechanism should ideally weigh the relevant feature heavily and conceal $\eta$.
However,  when the noise is too high, it becomes impossible to distinguish it from the relevant signal.
Thus, we define a threshold $\lambda$ for the noise level we can handle. 
Formally:


\begin{assumption}
\label{assumption:bound_on_noise}

There exists a system parameter $\lambda = \max_{r \sim r'} \lVert r.\eta - r'.\eta \rVert$.

\end{assumption}

With the right $\lambda$, $\epsilon$-Ordering Equality can be satisfied.



\section{Background: Differential Privacy}



Some readers may experience a  slight feeling of d\'eja vu by this point, as our definitions resemble those from a completely different realm -- differential privacy (DP). This is by design: the key contribution of this  paper is inlaying down the connection between removing bias for ordering fairness and differential privacy.

Differential privacy is a method to measure the information leakage of an algorithm. It assesses how much the algorithm’s output varies when its input is modified. If the output does not change significantly, the information about the input contained in the output is limited.

The goal of differential privacy mechanisms is to limit the information leakage of an algorithm. Enforcing fairness raises similar issues.
While we acknowledge that it may be impossible to  prevent $\eta$ from influencing the final score, we wish to guarantee that  $\eta$ does not affect the final ordering -- akin to the concern that  motivates differential privacy.

Before we proceed to explore further the connection between fair ordering and differential privacy, we quickly review below DP's
formal definition.

Suppose we have an algorithm $M : \mathcal{X}^n \rightarrow \mathcal{Y}$. Consider any two datasets $X, X' \in \mathcal{X}^n$, which differ in exactly one entry. We call these neighbouring datasets, and sometimes denote this by $X \sim X'$. 

\begin{definition}[$\epsilon$-Differential Privacy]
$M$ is \textit{$\epsilon$ differentially private} if, for all neighbouring $X, X'$, and all $T \subseteq \mathcal{Y}$, we have

\[ \text{Pr} [M(X) \in T] \leq e^{\epsilon} \text{Pr}[M(X') \in T] \]

where the randomness is over the choices made by M.

\end{definition}


Among the many existing implementations of $\epsilon$-differentially private algorithms, we focus on additive noise mechanisms. Commonly used for numerical values, these mechanisms rely on some notion of sensitivity to determine the appropriate amount of noise to add. One example is the $\ell_1$-sensitivity used in the \emph{Laplace mechanism}, a simple technique that ensures privacy by adding noise from the Laplace distribution.

\begin{definition}
Let $f: \mathcal{X}^n \rightarrow \mathbb{R}$.
The $\ell_1$-sensitivity of $f$ is:

\[\Delta^{(f)} = \max_{X \sim X'} \left\Vert f(X) - f(X') \right\Vert\]

    
\end{definition}

At a high level $\Delta^{(f)}$ quantifies how sensitive $f$ is to a single input -- {\em i.e.}, how much the output of the function can change when a single element of the input dataset is modified.













\section{Bridging Fairness and Privacy}
\subsection{$\epsilon$-Ordering Equality and Differential Privacy}
\label{subsec:eps}


To reveal the unexpected link between fairness and differential privacy, we demonstrate that any algorithm $\mathcal{A}$, which orders based on the \emph{score} assigned by the server and guarantees $\epsilon$-differential privacy, also satisfies $\epsilon$-Ordering Equality.



\begin{theorem}
\label{thm:equivalence}

    Let \emph{score} be a function that, for a given request $r$, returns the value $r.\text{relev} + r.\eta$, and let $\mathcal{D}$ be an additive-noise DP mechanism. If an algorithm $\mathcal{A}$ orders requests based on \emph{score}, and $\mathcal{D}$ is applied to $\mathcal{A}$ with sensitivity $\Delta = \lambda$, then $\mathcal{A}$ guarantees $\epsilon$-Ordering Equality.

\end{theorem}




\begin{proof}

Let $r,r'$ be two requests such that $r\sim r'$. By Assumption~\ref{assumption:bound_on_noise}, there exists a bound $\lambda$ such that $\lambda$ is the $\ell_1$-sensitivity of $score$. Since $\mathcal{A}$ is $\epsilon$-differentially private, for some fixed $s \in \mathcal{A}(score(r))$, we have
\begin{align*}
\text{Pr} [\mathcal{A}(score(r))=s] &\leq e^{\epsilon} \text{Pr}[\mathcal{A}(score(r'))=s] \\
\text{Pr} [\mathcal{A}(score(r'))=s] &\leq e^{\epsilon} \text{Pr}[\mathcal{A}(score(r))=s]
\end{align*}
\noindent Assume \[\text{Pr} [\mathcal{A}(score(r))=s] \leq \text{Pr}[\mathcal{A}(score(r'))=s]\]
We want to prove:
\[
    Pr[r\prec r'] \leq e^{\epsilon} Pr[r'\prec r]
\]
\[
    Pr[r'\prec r] \le e^{\epsilon} Pr[r\prec r']
\]
We have
\begin{align*}
Pr[r\prec r'] &= \sum_{s, s'}(kPr(\mathcal{A}(score(r)) = s)Pr(\mathcal{A}(score(r')) = s')) \\
&\leq \sum_{s, s'}(k Pr(\mathcal{A}(score(r')) = s) e^{\epsilon}Pr(\mathcal{A}(score(r)) = s')) \\
&=e^{\epsilon}\sum_{s, s'}(kPr(\mathcal{A}(score(r')) = s)Pr(\mathcal{A}(score(r)) = s')) \\
&=e^{\epsilon}Pr[r' \prec r]
\end{align*}
and
\begin{align*}
Pr[r' \prec r] &= \sum_{s, s'}(kPr(\mathcal{A}(score(r')) = s)Pr(\mathcal{A}(score(r)) = s')) \\
&\leq \sum_{s, s'}(k e^{\epsilon} Pr(\mathcal{A}(score(r)) = s)Pr(\mathcal{A}(score(r')) = s')) \\
&=e^{\epsilon}\sum_{s, s'}(kPr(\mathcal{A}(score(r)) = s)Pr(\mathcal{A}(score(r')) = s')) \\
&=e^{\epsilon}Pr[r \prec r']
\end{align*} satisfying $\epsilon$-Ordering Equality. If \[\text{Pr} [\mathcal{A}(score(r))=s] \geq \text{Pr}[\mathcal{A}(score(r'))=s]\] we still have $\epsilon$-Ordering Equality by symmetry.

\end{proof}

\subsection{$k\epsilon$-Ordering Equality and Group Privacy}

A notable limitation of $\epsilon$-Ordering Equality is that it only considers requests with identical relevant features.
In practice, however, it is crucial to handle requests with similar but not identical features, as slight variations, {\em e.g.}, due to measurement inaccuracies, could  be exploited to bypass fairness.
For example, a front-running or sandwich attack~\cite{daian2020flash,torres2021frontrunner} is unlikely to be directly addressed by $\epsilon$-Ordering Equality, as an adversary must learn about a request to bypass it, making it impossible for them to be issued at the \emph{exact} same time.


To overcome these limitations, we introduce a more robust and comprehensive version of $\epsilon$-Ordering Equality: $\boldsymbol{k\epsilon}$\textbf{-Ordering Equality}. This new definition retains the original guarantees for requests with identical relevant features, but extends support to cases when requests have similar features. It ensures that the probability of ordering two requests with similar features remains nearly identical, with the difference in probabilities increasing as the features diverge.

To express this flexibility, we introduce  the parameter
 $k = \frac{|\emph{score}(r) - \emph{score}(r')|}{\lambda}$ for any two requests $r,r'$.
We rely on the noise system $\lambda$ to normalize the difference in scores for any two requests. This difference affects directly the difference in probabilities. 

\begin{definition}[$k\epsilon$-Ordering Equality]
    For any two requests $r,r'$ it holds that $Pr[r\prec r']\leq 
 e^{k\epsilon} Pr[r'\prec r]$.

\end{definition}

For two adjacent requests, $k \leq 1$ and $k\epsilon$-Ordering Equality is an equivalent or tighter bound than $\epsilon$-Ordering Equality.
Having introduced an extended notion of fairness, we aim to ensure that our equivalence result remains valid. Indeed, this new definition still corresponds to a DP concept synonymous with $k\epsilon$-differential privacy -- \emph{group privacy}. 
In classic DP, privacy is guaranteed for databases that differ by only one object, often referred to as being at ``distance 1'' from each other. In our context, this corresponds to requests whose relevant features have the same values. Group privacy, however, extends the classic notion to databases that are at ``distance $k$'', {\em i.e.},  they differ by  $k$ distinct objects.
 For us, this maps to requests where the relevant features' values are similar but not identical.
 Using similar proof techniques as in \S\ref{subsec:eps}, we prove that group privacy ensures the preservation of order fairness, resulting in the following theorem:
 

\begin{theorem}
\label{thm:equivalence2}
    Let \emph{score} be a function that, for a given request $r$, returns the value $r.\text{relev} + r.\eta$, and let $\mathcal{D}$ be an additive-noise group privacy mechanism. If an algorithm $\mathcal{A}$ orders requests based on \emph{score}, and $\mathcal{D}$ is applied to $\mathcal{A}$ with sensitivity $\Delta = \lambda$, then $\mathcal{A}$ guarantees $k\epsilon$-Ordering Equality.

\end{theorem}

\subsection{Practical Considerations}
The parameter $\epsilon$ plays a similar role in our property as in differential privacy: a smaller $\epsilon$ introduces more noise, enhancing fairness but diminishing the influence of relevant features. Conversely, a larger $\epsilon$ reduces noise, prioritizing the values of relevant features at the expense of fairness. 
Companies like Google and Apple use $\epsilon \approx 2$ in DP~\cite{nist_diff_privacy_2023}, though a lower value might be ideal for fairness. The choice of $\epsilon$ should ultimately depend on the specific use case and desired trade-offs.

The Laplace mechanism is an easy-to-implement $\epsilon$-DP mechanism; however, it samples from an unbounded, double-sided domain.
In certain applications, this can cause liveness issues, making a bounded noise mechanism more suitable~\cite{holohan2018bounded}.


\section{Application for Blockchains}
\label{sec:DP_for_miners_fees}


\paragraph{Miner Fees and Bribery}

The first application examines the role of miner fees in transaction ordering. Here, the ordering server (e.g., the consensus leader) receives both a protocol fee and a per-transaction fee from clients. However, clients can also make side payments to prioritize their transactions, leading to MEV issues~\cite{daian2020flash}, where malicious leaders reorder transactions for profit.

To ensure fairness, we employ Ordering Equality, which deters side payments. The relevant feature here is the client fee, while out-of-band bribes are considered irrelevant as long as they remain bounded. We assume clients’ bribes do not exceed a certain threshold, reflecting their maximum willingness to pay for preferential treatment. Additionally, we assume client fees are significantly higher than any bribes, which can be justified in a two ways: ($i$) clients may prefer to trust the network for transaction inclusion rather than resorting to potentially shady side payments; and ($ii$) the entities paying fees and those offering bribes may differ. For instance, a client might pay a fee to access a decentralized application, while the application itself may have its own incentives to prioritize its own transaction on the global chain. In such cases, it makes sense that the fees would be higher than the bribes.

With these assumptions, we can straightforwardly apply ordering equality. By utilizing the DP mechanisms on the sum of value received by the leader, we ensure Ordering Equality, promoting fairness in transaction processing and reducing the potential for manipulation.

\paragraph{Time of Transaction Issuance.}
The second application considers the time of transaction issuance as the relevant feature, with network delays as noise. The server orders requests based on a score that combines the transaction's issuance time and the time taken to reach the server. This approach, detailed in~\cite{bercow}, effectively mitigates risks like sandwich attacks and front-running. Our method enhances this by more comprehensively addressing the nuances of relevant features differences.
\section{Related Work}

Several works have addressed fairness in SMR, proposing methods to mitigate biases in transaction ordering. Aequitas~\cite{kelkar2020order} introduced receive-order-fairness and $\gamma$-batch order-fairness. Since strict receive-order-fairness requires strong synchrony assumptions, Aequitas focused on $\gamma$-batch order-fairness, which ensures a transaction is ordered no later than another if received by a $\gamma$ fraction of correct replicas beforehand.

Themis~\cite{kelkar2023themis} builds on Aequitas by addressing liveness issue and providing total ordering with batch-order fairness. However, fairness is guaranteed only at the batch level, leaving vulnerabilities such as front-running within individual batches. Themis also requires $4f+1$ resilience, which is sub-optimal for SMR. Our approach is independent of SMR resilience models, utilizing Zhang et al.'s~\cite{bercow} secret random oracle (SRO) for common random variable sampling.

Quick Order Fairness~\cite{cachin2022quick} improves on Themis by offering differential order fairness, maintaining similar guarantees with better algorithmic efficiency.
Recent work~\cite{ramseyer2024brief} extended batch-order fairness to $(\gamma, \delta)$-minimal-batch-order-fairness, capturing the effects of faulty replicas misreporting transaction order.
Unlike their use of social choice theory, we approach fairness through the lens of databases and privacy, providing a more relevant foundation for addressing bias in SMR.

Wendy~\cite{kursawe2020wendy} and Pompe~\cite{zhang2020byzantine} introduce fair separability, which guarantees fair ordering if all honest replicas receive transactions in the same order. 
These approaches provide better guarantees than Aequitas and Themis but remain vulnerable to biases and manipulations by Byzantine clients, particularly due to sensitivity to irrelevant transaction features.

Zhang et al. attempts to address these vulnerabilities but but does not fully succeed, limiting its definition of relevant features to the time of issue. While this feature is the reasonable one to consider when targeting attacks like front-running, it fails when transactions have similar relevant features.
We extend their definition, using decaying probabilities to reorder transactions based on feature differences, specifically where irrelevant features might influence preferences by summing relevant features' values and noise. Addressing other types of noise remains an open question.

Ordering is a critical concern in blockchain systems. While the previous works reflect these concerns, they are not tested on widely-used systems. In contrast, deployed systems focus on mitigating unfairness using economic mechanisms~\cite{babel2024prof,flashbots2021overview,heimbach2023ethereum}, acknowledging that ordering will inevitably have some level of unfairness. 
We hope our work will inspire system designers to adopt our ideas and improve fairness in practical implementations.



\section{Conclusion}

This paper introduces $k\epsilon$-Ordering Equality, a notion of fair ordering that addresses two important limitation in the  framenwork for fair ordering  introduced by Zhang et al.~\cite{bercow}. Further, it establishes  an unexpected connection between this notion of fair ordering and differential privacy. 
We are intrigued by the possibility of a deeper, more fundamental connection between privacy and fairness, which we hope to explore in future work.



\bibliographystyle{plainurl}
\bibliography{main}

\newpage
\begin{appendix}
\end{appendix}

\end{document}